\documentclass[11pt,english]{article}
\usepackage{lmodern}

\usepackage[T1]{fontenc}
\usepackage[latin9]{inputenc}
\usepackage{geometry}
\usepackage{amsthm}
\usepackage{amsmath}
\usepackage{amssymb}
\usepackage{graphicx}

\makeatletter
\theoremstyle{plain}
\newtheorem{thm}{\protect\theoremname}
  \theoremstyle{definition}
  \newtheorem{defn}{\protect\definitionname}
  \theoremstyle{remark}
  \newtheorem{rem}{\protect\remarkname}
  \theoremstyle{plain}
  \newtheorem{lem}{\protect\lemmaname}
  \theoremstyle{plain}
  \newtheorem{cor}{\protect\corollaryname}
  \theoremstyle{plain}
  \newtheorem{fact}{\protect\factname}

\usepackage{graphicx}
\usepackage{algpseudocode}
\date{}

\makeatother

\usepackage{babel}
  \providecommand{\definitionname}{Definition}
  \providecommand{\factname}{Fact}
  \providecommand{\lemmaname}{Lemma}
  \providecommand{\remarkname}{Remark}
\providecommand{\corollaryname}{Corollary}
\providecommand{\theoremname}{Theorem}

\begin{document}

\title{Counting perfect matchings in graphs that exclude a single-crossing
minor}

\author{Radu Curticapean\thanks{Saarland University, Dept.of Computer Science, \tt{curticapean@cs.uni-sb.de}}}

\maketitle
\global\long\def\PerfMatch{\mathrm{PerfMatch}}
\global\long\def\sharpP{\mbox{\ensuremath{\mathsf{\#P}}}}
\global\long\def\Sig{\mathrm{Sig}}
\global\long\def\bigoh{\mathcal{O}}
\global\long\def\PM{\mathcal{PM}}
\global\long\def\Match{\mathcal{M}}
\global\long\def\T{\mathcal{T}}

\begin{abstract}
A graph $H$ is \emph{single-crossing} if it can be drawn in the plane
with at most one crossing. For any single-crossing graph $H$, we
give an $\mathcal{O}(n^{4})$ time algorithm for counting perfect
matchings in graphs excluding $H$ as a minor. The runtime can be
lowered to $\mathcal{O}(n^{1.5})$ when $G$ excludes $K_{5}$ or
$K_{3,3}$ as a minor.

This is the first generalization of an algorithm for counting perfect
matchings in $K_{3,3}$-free graphs (Little 1974, Vazirani 1989).
Our algorithm uses black-boxes for counting perfect matchings in planar
graphs and for computing certain graph decompositions. Together with
an independent recent result (Straub et al. 2014) for graphs excluding
$K_{5}$, it is one of the first nontrivial algorithms to not inherently
rely on Pfaffian orientations.
\end{abstract}

\section{Introduction}

A \emph{perfect matching} of a graph $G=(V,E)$ is a set $M\subseteq E$
of $|V|/2$ vertex-disjoint edges. For an edge-weighted graph $G$
with weights $w:E\to\mathbb{Q}$, we consider the problem of computing
$\PerfMatch(G)=\sum_{M}\prod_{e\in M}w(e)$, where the outer sum ranges
over all perfect matchings $M$ of $G$. If $w(e)=1$ for all $e\in E(G)$,
this quantity plainly counts perfect matchings of $G$.

The problem $\PerfMatch$ arises in statistical physics as the dimer
problem \cite{Kasteleyn19611209,FisherTemperley}. In algebra and
combinatorics, the quantity $\PerfMatch(G)$ for bipartite $G$ is
better known as the permanent of the (bi-)adjacency matrix of $G$.
The complexity of its evaluation is of central interest in counting
complexity \cite{DBLP:journals/tcs/Valiant79} and algebraic complexity
\cite{B-2000-Completeness-and-Reduction-in-Algebraic-Complexity-Theory}.
In fact, the permanent was the first natural problem with a polynomial-time
decision version that was shown $\sharpP$-hard, even for zero-one
weights, thus demonstrating that counting can be harder than decision.

To cope with this hardness, several reliefs were proposed: If counting
may be relaxed to approximate counting, then the problem becomes feasible:
It was shown in \cite{DBLP:journals/jacm/JerrumSV04} that $\PerfMatch(G)$
admits a fully polynomial randomized approximation scheme on graphs
$G$ with non-negative edge weights. If the exact value of $\PerfMatch(G)$
is required, but $G$ may be restricted to a specific class of graphs,
then a rather short list of polynomial-time algorithms is known:

For planar $G$, the value $\PerfMatch(G)$ can be computed in time
$\mathcal{O}(n^{1.5})$ by \cite{FisherTemperley,Kasteleyn19611209}.
Interestingly, this algorithm from 1967 predates the hardness result
for general graphs. Note that planar graphs exclude both $K_{3,3}$
and $K_{5}$ as a minor. In \cite{PM_Little,DBLP:journals/iandc/Vazirani89},
the previous algorithm was generalized to a (parallel) algorithm on
graphs $G$ that are only required to exclude the minor $K_{3,3}$.
Orthogonally to this, it was shown in \cite{DBLP:journals/combinatorics/GalluccioL99}
that $\PerfMatch(G)$ admits an $\mathcal{O}(4^{g}n^{3})$ algorithm
on graphs that can be embedded on a surface of genus $g$. Recently,
and independently of this work, a (parallel) polynomial-time algorithm
was shown in \cite{ThieraufPM} for computing $\PerfMatch(G)$ on
graphs excluding $K_{5}$ as a minor. In the present paper, we show:
\begin{thm}
\label{thm:main}Let $H$ be a single-crossing graph, i.e., $H$ can
be drawn in the plane with at most one crossing. Then there is an
$\mathcal{O}(n^{4})$ time algorithm for computing $\PerfMatch(G)$
on input graphs $G$ that exclude $H$ as a minor. If $H$ is one
of the single-crossing graphs $K_{5}$ or $K_{3,3}$, then the runtime
can be lowered to $\mathcal{O}(n^{1.5})$.
\end{thm}
Note that the excluded minor $H$, rather than $G$, is required to
be single-crossing: Algorithms for single-crossing $G$ would follow
from a very simple reduction to the planar case. 

Theorem~\ref{thm:main} directly generalizes the algorithm for graphs
excluding $K_{3,3}$ or $K_{5}$, but is orthogonal to the result
for bounded-genus graphs: The graph consisting of $n$ disjoint copies
of the single-crossing graph $K_{5}$ has genus $\Theta(n)$, but
excludes $K_{3,3}$ as a minor. Thus, Theorem~\ref{thm:main} applies
on this graph, while the algorithm for bounded-genus graphs does not.
Conversely, the class of torus-embeddable graphs includes all single-crossing
graphs. Thus, the algorithm for bounded-genus graphs applies here,
while Theorem~\ref{thm:main} does not.

Graphs excluding a single-crossing minor $H$ have already been studied:
By a decomposition theorem \cite{DBLP:conf/gst/RobertsonS91}, which
constitutes a fragment of the general graph structure theorem for
general $H$-minor free graphs \cite{Robertson200343}, such graphs
can be decomposed into planar graphs and graphs of bounded treewidth,
and it was shown in \cite{DBLP:journals/jcss/DemaineHNRT04} how to
compute such decompositions. Furthermore, approximation algorithms
for the treewidth and other invariants of such graphs are known \cite{DBLP:journals/jcss/DemaineHNRT04,DBLP:conf/approx/DemaineHT02},
as well as $\mathcal{O}(n\log n)$ algorithms for computing maximum
flows \cite{DBLP:journals/jgaa/ChambersE13}.

Our algorithm requires black-boxes for $\PerfMatch$ on planar graphs
and for finding the decompositions described above. We also use the
concept of matchgates from \cite{DBLP:journals/siamcomp/Valiant08},
but can limit ourselves to a self-contained fragment of their theory.
All required ingredients are introduced in Section~\ref{sec:Preliminaries}
and used in Section~\ref{sub:Algorithm} to present the algorithm
proving Theorem~\ref{thm:main}.

\section{\label{sec:Preliminaries}Mise en place}

Let $\mathbb{F}$ be a field supporting efficient arithmetic operations.
Graphs $G=(V,E)$ are undirected and may feature parallel edges and
weights $w:E\to\mathbb{F}$. We allow zero-weight edges $e\in E$
with $w(e)=0$ and write $|G|:=|V(G)|$.

A graph $G$ is planar if it admits an embedding $\pi$ into the plane
without crossings, and single-crossing if it admits an embedding into
the plane with at most one crossing. Examples for single-crossing
graphs are $K_{5}$ and $K_{3,3}$. A plane graph is a pair $(G,\pi)$,
where $\pi$ is a planar embedding of $G$. Given a plane graph $(G,\pi)$
and a cycle $C$ in $G$, we say that $C$ bounds a face in $G$ if
one of the regions bounded by $C$ in $\pi$ is empty.

We write $\PM[G]$ for the set of perfect matchings of $G$ and define
$w(M)=\prod_{e\in M}w_{G}(e)$ and $\PerfMatch(G)=\sum_{M\in\PM[G]}w(M)$.
As already noted, despite its $\sharpP$-hardness on general graphs,
the value $\PerfMatch(G)$ can be computed in polynomial time for
planar $G$. 
\begin{thm}
\label{thm:planar}For planar graphs $G$, the value $\PerfMatch(G)$
can be computed in time $\mathcal{O}(n^{1.5})$.\end{thm}
\begin{proof}
(Sketch of \cite{Kasteleyn19611209}) In time $\mathcal{O}(n)$, we
can compute a set $S\subseteq E(G)$ such that the following holds:
After flipping the sign of $w(e)$ for each edge $e\in S$, we obtain
a new planar graph with adjacency matrix $A'$ satisfying $\PerfMatch(G)=\sqrt{\det(A')}$.
If $A'$ is the adjacency matrix of a planar graph, then $\det(A')$
can be computed in time $\mathcal{O}(n^{1.5})$ by \cite{Dissection},
noted also in \cite{DBLP:journals/siamcomp/Valiant08}.
\end{proof}

\subsection{Graph minors and decompositions}

A graph $H$ is a minor of $G=(V,E)$ if $H$ can be obtained from
$G$ by repeated edge/vertex-deletions and edge-contractions. The
contraction of $uv\in E$ identifies vertices $u,v\in V(G)$ to a
new vertex $w$ and replaces possible edges $uz\in E$ or $vz\in E$
for $z\in V(G)$ by a new edge $wz$. For a graph class $\mathcal{H}$,
write $\mathcal{C}[\mathcal{H}]$ for the class of all graphs $G$
such that no $H\in\mathcal{H}$ is a minor of $G$. By Kuratowski's
theorem, $\mathcal{C}[K_{3,3},K_{5}]$ coincides with the planar graphs. 

Other graph classes can also be expressed by forbidden minors. In
fact, Robertson and Seymour's graph structure theorem \cite{Robertson200343}
describes the structure of graphs in $\mathcal{C}[H]$ for arbitrary
$H$. We use a fragment of this theorem that applies only when $H$
is single-crossing: Roughly speaking, graphs in $\mathcal{C}[H]$
consist of planar graphs and constant-size graphs that are glued together
in a well-specified way. Our algorithm will crucially rely on these
decompositions.
\begin{defn}
\label{def:decomp}
\begin{figure}[t]
\begin{centering}
\includegraphics[width=0.87\textwidth]{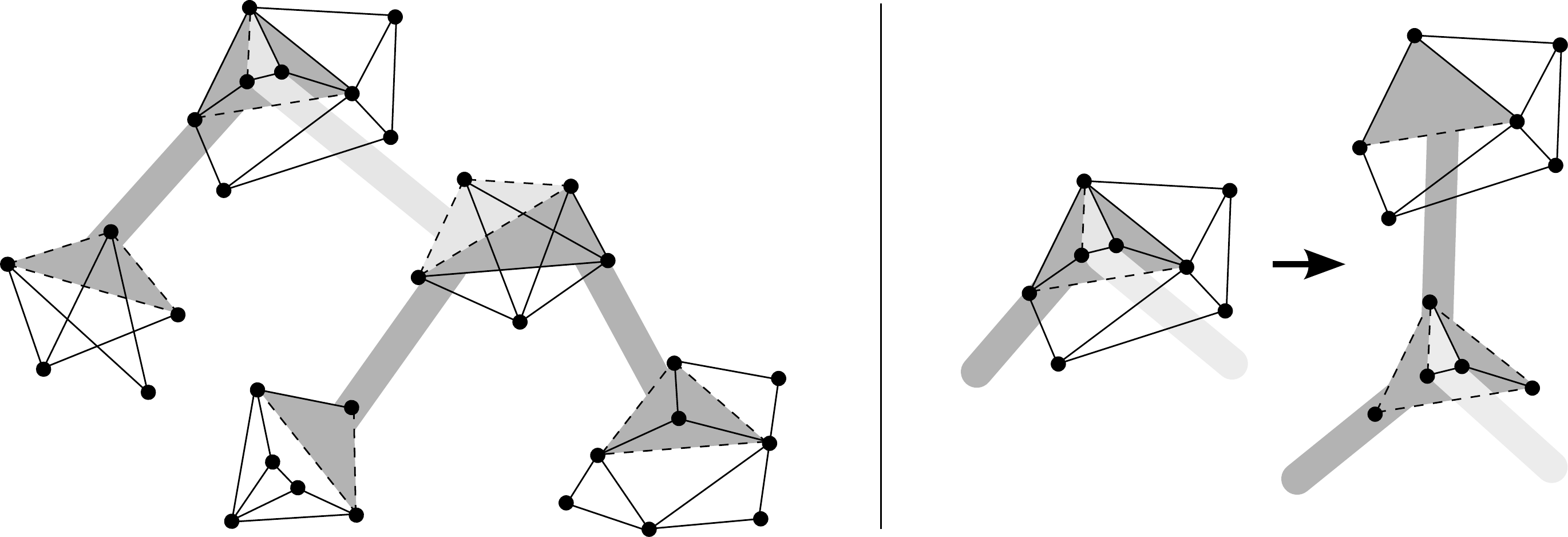}
\par\end{centering}

\caption{\label{fig:decomp}(left) $\mathcal{T}$ is almost $5$-nice: Either
$|V(G{}_{t})|\leq5$ or $G_{t}$ is a plane graph whose non-navel
attachment cliques bound faces, with the exception of one triangle
$K$ at the root. Zero-weight edges are drawn with dashed lines. (right)
The offending attachment clique $K$ is repaired.}
\end{figure}
Let $F,F'$ be graphs, both containing a vertex set $K$. Write $F\oplus_{K}F'$
for the graph obtained from the disjoint union of $F$ and $F'$ by
identifying, for each $v\in K$, the two copies of $v$. This may
create parallel edges between vertices in $K$. \end{defn}
\begin{itemize}
\item In the following, let $G$ be a graph.\textbf{ }A \emph{decomposition}
$\mathcal{T}=(T,\mathcal{G})$ of $G$ is a rooted tree $T$ with
a family of graphs $\mathcal{G}=\{G_{t}\}_{t\in V(T)}$ such that
the following holds:

\begin{enumerate}
\item For $st\in E(T)$, the set $K[s,t]:=V(G_{s})\cap V(G_{t})$ is a clique,
the so-called \emph{attachment clique} at $st$, possibly containing
zero-weight edges in $G_{s}$ or $G_{t}$. If $s$ is the parent of
$t$, we call $K[s,t]$ the \emph{navel} of $t$.
\item For $t\in V(T)$, define $G_{\leq t}$: If $t$ is a leaf, then $G_{\leq t}=G_{t}$.
If $t$ has children $s_{1},\ldots,s_{r}$ with navels $K_{1},\ldots,K_{r}$,
then $G_{\leq t}=G_{t}\oplus_{K_{1}}G_{\leq s_{1}}\oplus_{K_{2}}\ldots\oplus_{K_{r}}G_{\leq s_{r}}$.
If $t$ is the root, we require that $G_{\leq t}$ is isomorphic to
$G$ after removal of all zero-weight edges.
\end{enumerate}
\item For $c\in\mathbb{N}$, the decomposition $\mathcal{T}$ is \emph{$c$-nice}
if $G_{t}$ is given as a plane graph whenever $|V(G_{t})|>c$. Furthermore,
if $K$ is an attachment clique in $G_{t}$, then $|K|\leq3$. If
$|K|=3$ and $K$ is not the navel of $G_{t}$, then $K$ is required
to bound a face in $G_{t}$.
\item If $|V(G_{t})|\leq k$ for all $t\in V(T)$, then $\mathcal{T}$ is
a \emph{tree-decomposition} of width $k$ of $G$. The \emph{treewidth}
of $G$ is defined as $\min\{k\in\mathbb{N}\mid G\mbox{ has a tree-decomposition of width }k+1\}$. \end{itemize}
\begin{rem}
\label{rem:clique-decomp}The above definition of treewidth, used
e.g.~in \cite{Kriz1990177}, is equivalent to the more common one
that uses ``bags''. It is also verified that, if $\mathcal{T}$
is a decomposition of $G$ and $K$ is a clique in $G$, then there
is some node $t$ in $\mathcal{T}$ such that $K\subseteq V(G_{t})$.\end{rem}
\begin{thm}
\label{thm:decomposition}For every single-crossing graph $H$, there
is a constant $c\in\mathbb{N}$ such that the following holds: For
every $G\in\mathcal{C}[H]$, a $c$-nice decomposition $\mathcal{T}=(T,\mathcal{G})$
of $G$ can be found in time $\mathcal{O}(n^{4})$. Additionally,
$\mathcal{T}$ satisfies the size bounds $\sum_{t\in V(T)}|G_{t}|\in\mathcal{O}(n)$
and $|T|\in\mathcal{O}(n)$.\end{thm}
\begin{proof}
Using the decomposition algorithm presented in \cite{DBLP:journals/jcss/DemaineHNRT04},
we compute in $\mathcal{O}(n^{4})$ time a decomposition $\mathcal{T}'=(T',\mathcal{G}')$
that satisfies the following: For each $t\in V(T')$, either $G_{t}$
has treewidth $\leq c$, or $G_{t}$ is a plane graph whose attachment
cliques $K$ satisfy $|K|\leq3$. Furthermore, $\mathcal{T}'$ satisfies
the size bounds stated in the theorem for $\mathcal{T}$.

By local patches at nodes $t\in V(T)$, we successively transform
$\mathcal{T}'$ to a $c$-nice decomposition $\mathcal{T}$. This
involves (i) splitting nodes $t$ of treewidth $\leq c$ into trees
of constant-size parts, and (ii) splitting planar nodes into multiple
planar nodes whose non-navel attachments bound faces. 

With $Z_{t}$ denoting the set of nodes added to $\mathcal{T}'$ by
patching $t$, we show along the way that the local size bound $\sum_{z\in Z_{t}}|G_{z}|\in\mathcal{O}(|G_{t}|)$
holds. This implies the claimed size bounds on $\mathcal{T}$.

\textbf{(i)} Let $G_{t}$ have treewidth $\leq c$. Using \cite{Bodlaender:1996:LAF:243705.243727},
compute in time $\mathcal{O}(2^{c^{3}}n)$ a tree-decomposition $\mathcal{R}=(R,\mathcal{B})$
of width $c$ of $G_{t}$ with $\mathcal{B}=\{B_{r}\}_{r\in V(R)}$
and $|R|\in\mathcal{O}(|G_{t}|)$. Let $K$ be the navel of $t$ and
let $r$ be an arbitrary node of $R$ satisfying $K\subseteq V(B_{r})$,
which exists by Remark~\ref{rem:clique-decomp}. Declare $r$ as
root of $\mathcal{R}$ and attach $\mathcal{R}$ to $\mathcal{T}'$
by deleting $t$ from $\mathcal{T}'$, disconnecting possible children
of $t$, and inserting $\mathcal{R}$ with root $r$ at the place
of $t$. For every child $s$ of $t$ in $\mathcal{T}'$ that was
disconnected this way, do the following: By Remark~\ref{rem:clique-decomp},
its navel, which is a clique, is contained in $B_{p}$ for some node
$p$ of $\mathcal{R}$. Add the edge $ps$ to $\mathcal{T}'$. Processing
$t$ this way adds $|R|\in\mathcal{O}(|G_{t}|)$ new nodes $z$ to
$\mathcal{T}'$, each with $|G_{z}|\leq c$, showing the local size
bound for $t$.

\textbf{(ii)} Similar to \cite{DBLP:journals/jgaa/ChambersE13}. Let
$K$ be an attachment clique of $G_{t}$ that does not bound a face,
as in Figure~\ref{fig:decomp}. Then $t$ has a neighbor $s$ such
that the subgraph $F$ bounded by $K=K[s,t]$ in the embedding of
$G_{t}$ contains other vertices than $K$. Delete $F-K$ from $G_{t}$.
Add a new node $t'$ adjacent to $t$ and define $G_{t'}:=F$ with
zero weight at all edges in $F[K]$. For each child $r$ of $t$ whose
navel is contained in $V(F)$, replace the edge $rt$ of $T$ by $rt'$.
If the newly created graph $G_{t'}$ contains another attachment clique
that does not bound a face, recurse on $G_{t'}$. 

For (ii), we see that $|Z_{t}|\leq|G_{t}|$ since every recursion
step deletes at least one vertex from its current subgraph of $G_{t}$.
Secondly, the local size bound holds at $t$ since every recursion
step introduces at most $3$ new vertices, namely the copy of $K$
in the child node.\end{proof}
\begin{rem}
\label{rem:runtime}For $H\in\{K_{3,3},K_{5}\}$, an $\mathcal{O}(1)$-nice
decomposition $\mathcal{T}$ can be found in time $\mathcal{O}(n)$:
Instead of computing $\mathcal{T}'$ by \cite{DBLP:journals/jcss/DemaineHNRT04}
in the first step, use \cite{Asano1985249} for $H=K_{3,3}$ or \cite{LATIN_K5}
for $H=K_{5}$.
\end{rem}

\subsection{Matchgates and signatures}

In the following, we present the concept of matchgates from \cite{DBLP:journals/siamcomp/Valiant08},
as these will play a central role in our algorithm. We limit ourselves
to a small self-contained fragment of their theory.
\begin{defn}
[\cite{DBLP:journals/siamcomp/Valiant08}] A \emph{matchgate} $\Gamma=(G,S)$
is a graph $G$ with a set of external vertices $S\subseteq V(G)$.
Its \emph{signature} $\Sig(\Gamma):2^{S}\to\mathbb{F}$ is the function
that maps $X\subseteq S$ to $\PerfMatch(G-X)$.\end{defn}
\begin{rem}
\label{rem:compute-matchgate}For $\Gamma=(G,S)$ with $|S|=k$, we
represent $\Sig(\Gamma)$ by a vector in $\mathbb{F}^{2^{k}}$. If
we can compute $\PerfMatch(G-X)$ for $X\subseteq S$ in time $t$,
then we can compute $\Sig(\Gamma)$ in time $\bigoh(2^{k}t)$.
\end{rem}
The signature of $\Gamma$ describes its behavior in sums with other
graphs:
\begin{lem}
\label{lem:sig-prod}For matchgates $\Gamma=(G,S)$ and $\Gamma'=(G',S)$,
let $G^{*}=G\oplus_{S}G'$. Then 
\begin{equation}
\PerfMatch(G^{*})=\sum_{Y\subseteq S}\Sig(\Gamma,Y)\cdot\Sig(\Gamma',S\setminus Y).\label{eq:join}
\end{equation}
\end{lem}
\begin{proof}
Each $M\in\PM[G^{*}]$ induces a unique partition into $M=N\cup N'$
with $N\subseteq E(G)$ and $N'\subseteq E(G')$. Since $M$ is a
perfect matching, every $v\in V(G^{*})$ is matched in exactly one
of $N$ or $N'$. For vertices $v\notin S$, the choice of $N$ or
$N'$ independent of $M$.

For $Y\subseteq S$, let $\mathcal{M}_{Y}\subseteq\PM[G^{*}]$ denote
the perfect matchings of $G^{*}$ with $S\setminus Y$ matched by
$N$ and $Y$ matched by $N'$. Since $\{\mathcal{M}_{Y}\}_{Y\subseteq S}$
partitions $\PM[G^{*}]$, we have $\PerfMatch(G^{*})=\sum_{Y\subseteq S}\sum_{M\in\mathcal{M}_{Y}}w(M)$.
It remains to show $\sum_{M\in\mathcal{M}_{Y}}w(M)=\Sig(\Gamma,Y)\cdot\Sig(\Gamma',S\setminus Y)$:
This follows since every $M\in\mathcal{M}_{Y}$ can be written as
$M=N\cup N'$ with $(N,N')\in\PM[G-Y]\times\PM[G'-(S\setminus Y)]$
and the correspondence between $M$ and $(N,N')$ is bijective.
\end{proof}
Since the only information used about $G'$ in (\ref{eq:join}) is
contained in $\Sig(\Gamma')$, we conclude:
\begin{cor}
\label{cor:sig-equiv}Let $\Gamma=(F,S)$ and $\Gamma'=(F',S)$ and
let $G$ be a graph with $S\subseteq V(G)$. If $\Sig(\Gamma)=\Sig(\Gamma')$,
then $\PerfMatch(G\oplus_{S}\Gamma)=\PerfMatch(G\oplus_{S}\Gamma')$.
\end{cor}
Whenever $\Gamma$ has $\leq3$ external vertices, we can find a small
planar matchgate $\Gamma'$ with the same signature. We show this
in the next fact, essentially from \cite{DBLP:journals/siamcomp/Valiant08}.
Together with Corollary~\ref{cor:sig-equiv}, we will use $\Gamma'$
to mimick $\Gamma$, similarly to an idea in \cite{DBLP:journals/jgaa/ChambersE13}
for mimicking flow networks.
\begin{fact}
\label{fact:matchgate} For every matchgate $\Gamma=(G,S)$ with $|S|\leq3$,
there is a matchgate $\Gamma'=(F,S)$ with $\Sig(\Gamma)=\Sig(\Gamma')$
such that $F$ is a plane graph on $\leq7$ vertices with $S$ on
its outer face.
\begin{figure}[t]
\begin{centering}
\includegraphics[width=1\textwidth]{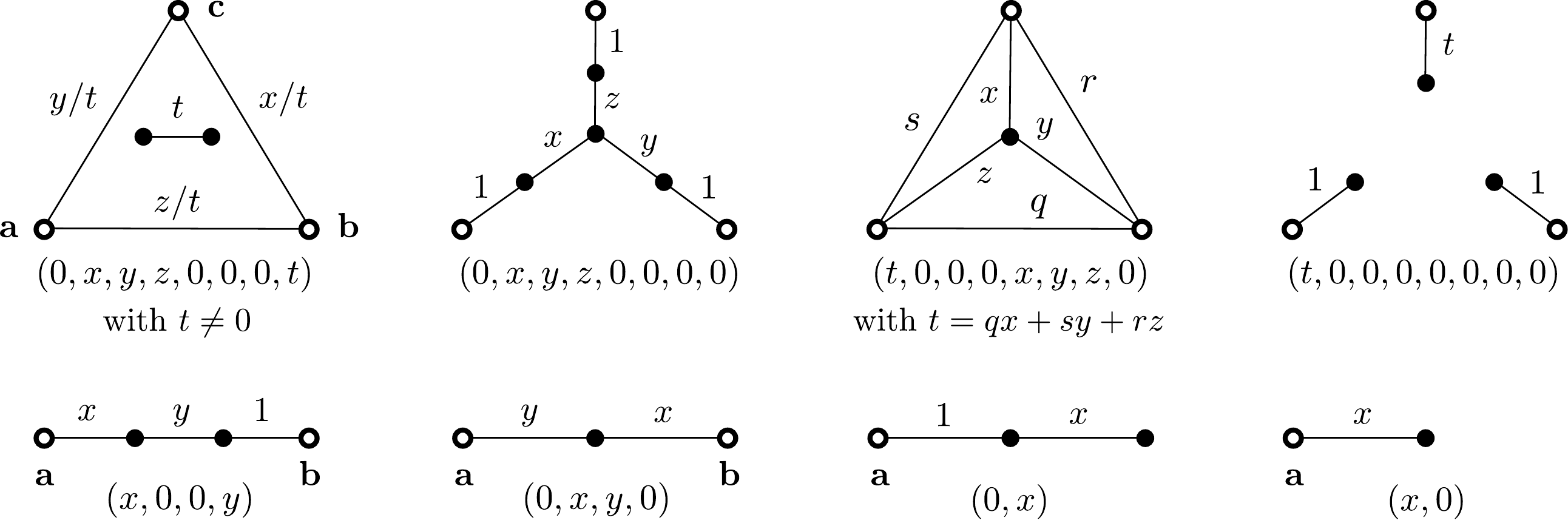}
\par\end{centering}

\caption{\label{fig:matchgates}The matchgates from Propositions 6.1 and 6.2
in \cite{DBLP:journals/siamcomp/Valiant08}, each drawn as a plane
graph with a set $S\subseteq\{\mathbf{a},\mathbf{b},\mathbf{c}\}$
as external vertices on the outer face. Below each matchgate, its
signature is given as a vector of length $2^{|S|}$ with entries ordered
as $\emptyset,a,b,c,ab,ac,bc,abc$ or a subsequence thereof. If $f$
is even or odd, then at least one matchgate $\Gamma$ satisfies $\Sig(\Gamma)=f$:
If $|S|=3$ and $f$ is even, then either the first or second matchgate
applies. If $|S|=3$ and $f$ is odd, the third or fourth matchgate
applies. If $|S|\leq2$, a matchgate of the second row applies.}
\end{figure}
\end{fact}
\begin{proof}
We call $f:2^{S}\to\mathbb{F}$ \emph{even} if $f(X)=0$ for all $X$
of odd cardinality, and we call $f$ \emph{odd} if $f(X)=0$ for all
$X$ of even cardinality. Since every matching features an even number
of matched vertices, $\Sig(\Gamma)$ is even/odd if $|G|$ is even/odd.
Hence Figure~\ref{fig:matchgates}, adapted from \cite{DBLP:journals/siamcomp/Valiant08},
contains a matchgate with signature $\Sig(\Gamma)$ after suitable
substitution of edge weights.
\end{proof}

\section{\label{sub:Algorithm}Proof of Theorem 1}

By Theorem~\ref{thm:decomposition}, if $G$ excludes a fixed single-crossing
minor $H$, we can find a $c$-nice decomposition $\mathcal{T}=(T,\mathcal{G})$
with $c\in\mathcal{O}(1)$. This $\mathcal{T}$ satisfies $\sum_{t\in V(T)}|G_{t}|\in\mathcal{O}(n)$
and $|T|\in\mathcal{O}(n)$.

For $t\in V(T)$, let $n_{t}=|G_{t}|$. For non-root nodes $t\in V(T)$
with navel $K$, define the matchgate $\Gamma_{\leq t}=(G_{\leq t},K)$.
For the root $r\in V(T)$, note that $G_{\leq r}=G$. Since $r$ has
no navel, write $\Gamma_{\leq r}=(G,\emptyset)$ by convention.

We compute $\Sig(\Gamma_{\leq t})$ for each $t\in V(T)$ by a bottom-up
traversal of $\mathcal{T}$. This computes $\Sig(\Gamma_{\leq r},\emptyset)$
for the root $r$, which is equal to $\PerfMatch(G)$ by definition.
To process $t\in V(T)$, we assume that $\Sig(\Gamma_{\leq r})$ is
known for each child $r$ of $t$. This is trivially true if $t$
is a leaf and will be assumed by induction for non-leaf nodes. We
then compute $\Sig(\Gamma_{\leq t})$ as follows:
\begin{itemize}
\item If $G_{t}$ has $\leq c$ vertices, let $V=V(G_{t})$, let $\Delta_{0}=(G_{t},V)$
and compute $\Sig(\Delta_{0})$ in time $2^{\mathcal{O}(c^{2})}$
by brute force. Let $s_{1},\ldots,s_{b}$ be the children of $t$,
with navels $K_{1},\ldots,K_{b}\subseteq V$. For $1\leq i\leq b$,
define $\Delta_{i}=(G_{t}\oplus_{K_{1}}G_{\leq s_{1}}\oplus_{K_{2}}\ldots\oplus_{K_{i}}G_{\leq s_{i}},V)$
and successively compute $\Sig(\Delta_{i})$ from the values of $\Sig(\Delta_{i-1})$
and $\Sig(G_{\leq s_{i}})$ by means of Lemma~\ref{lem:sig-prod}
and Remark~\ref{rem:compute-matchgate}. After completing this, since
the external nodes $V$ of $\Delta_{b}$ trivially include the navel
of $t$, we obtain $\Sig(\Gamma_{\leq t})$ as a restriction of $\Sig(\Delta_{b})$.
\item If $G_{t}$ is planar, first perform the following for each attachment
clique $K$ of $G_{t}$: 

\begin{enumerate}
\item Let $s_{1},\ldots,s_{b}$ denote the children of $t$ with navel $K$
and define the matchgate $\Delta=(G_{\leq s_{1}}\oplus_{K}\ldots\oplus_{K}G_{\leq s_{b}},K)$.
Recall that $|K|\leq3$ since $\mathcal{T}$ is nice.
\item Use Lemma~\ref{lem:sig-prod} to compute $f=\Sig(\Delta)$ and use
Fact~\ref{fact:matchgate} to obtain a planar matchgate $\Phi$ on
external vertices $K$ with $\Sig(\Phi)=f$ and $K$ on its outer
face. 
\item Replace $G_{t}$ by $G_{t}\oplus_{K}\Phi$, resulting in a planar
graph: Planarity is obvious if $|K|\leq2$. If $|K|=3$, recall that
$K$ lies on the outer face of $\Phi$, and that $K$ bounds a face
in $G_{t}$. The union of such planar graphs preserves planarity.
\end{enumerate}

After processing all attachment cliques, the graph $G_{t}$ is planar
and has $\mathcal{O}(n_{t})$ vertices. By Corollary~\ref{cor:sig-equiv},
we have $\Sig(\Psi)=\Sig(\Gamma_{\leq t})$ for $\Psi=(G_{t},K)$,
where $K$ with $|K|\leq3$ is the navel of $t$. Compute $\Sig(\Psi)$
by Theorem~\ref{thm:planar} and Remark~\ref{rem:compute-matchgate}
in time $\mathcal{O}(n_{t}^{1.5})$.

\end{itemize}
By Theorem~\ref{thm:decomposition} and Remark~\ref{rem:runtime},
computing $\mathcal{T}$ requires $\mathcal{O}(n^{4})$ time for general
$H$ or $\mathcal{O}(n)$ time for $H\in\{K_{3,3},K_{5}\}$. Processing
$\mathcal{T}$ requires time $\mathcal{O}(|T|+\sum_{t\in T}n_{t}^{1.5})$:
At node $t$, we spend either $2^{\mathcal{O}(c^{2})}$ or $\mathcal{O}(n_{t}^{1.5})$
time. Since $\sum_{t\in T}n_{t}\in\mathcal{O}(n)$ by the size bound
of Theorem~\ref{thm:decomposition}, it follows that $\sum_{t\in T}n_{t}^{1.5}\leq(\sum_{t\in T}n_{t})^{1.5}\in\mathcal{O}(n^{1.5})$.
As $|T|\in\mathcal{O}(n)$, the overall runtime claims follow.

\section{Conclusions and future work}

We presented a polynomial-time algorithm for $\PerfMatch(G)$ on graphs
$G\in\mathcal{C}[H]$ when $H$ is single-crossing. Since structural
results about graphs in $\mathcal{C}[H]$ for arbitrary (and not necessarily
single-crossing) graphs $H$ are known \cite{Robertson200343}, it
is natural to ask whether our approach can be extended to such graphs.
We cautiously believe in an affirmative answer -- in fact, Mingji
Xia and the author made some progress towards a proof, but are still
facing nontrivial obstacles.

\begingroup \fontsize{10pt}{11pt}\selectfont

\bibliographystyle{plain}
\bibliography{1CR}

\begin{thebibliography}{10}

\bibitem{Asano1985249}
T.~Asano.
\newblock An approach to the subgraph homeomorphism problem.
\newblock {\em Theor. Comp. Sci.}, 38(0):249--267, 1985.

\bibitem{Bodlaender:1996:LAF:243705.243727}
H.~Bodlaender.
\newblock A linear-time algorithm for finding tree-decompositions of small
  treewidth.
\newblock {\em SIAM J. Comput.}, 25(6):1305--1317, December 1996.

\bibitem{B-2000-Completeness-and-Reduction-in-Algebraic-Complexity-Theory}
P.~B\"urgisser.
\newblock {\em Completeness and Reduction in Algebraic Complexity Theory}.
\newblock Number~7 in Algorithms and Computation in Mathematics. Springer
  Verlag, 2000.
\newblock 168 + xii pp.

\bibitem{DBLP:journals/jgaa/ChambersE13}
E.~Chambers and D.~Eppstein.
\newblock Flows in one-crossing-minor-free graphs.
\newblock {\em J. Graph Algorithms Appl.}, 17(3):201--220, 2013.

\bibitem{DBLP:journals/jcss/DemaineHNRT04}
E.~Demaine, M.~Hajiaghayi, N.~Nishimura, P.~Ragde, and D.~Thilikos.
\newblock Approximation algorithms for classes of graphs excluding
  single-crossing graphs as minors.
\newblock {\em J. Comput. Syst. Sci.}, 69(2):166--195, 2004.

\bibitem{DBLP:conf/approx/DemaineHT02}
E.~Demaine, M.~Hajiaghayi, and D.~Thilikos.
\newblock 1.5-approximation for treewidth of graphs excluding a graph with one
  crossing as a minor.
\newblock In {\em APPROX}, pages 67--80, 2002.

\bibitem{DBLP:journals/combinatorics/GalluccioL99}
A.~Galluccio and M.~Loebl.
\newblock On the theory of pfaffian orientations. {I.} {P}erfect matchings and
  permanents.
\newblock {\em Electr. J. Comb.}, 6, 1999.

\bibitem{DBLP:journals/jacm/JerrumSV04}
M.~Jerrum, A.~Sinclair, and E.~Vigoda.
\newblock A polynomial-time approximation algorithm for the permanent of a
  matrix with nonnegative entries.
\newblock {\em J. ACM}, 51(4):671--697, 2004.

\bibitem{Kasteleyn19611209}
P.~Kasteleyn.
\newblock Graph theory and crystal physics.
\newblock In {\em Graph Theory and Theoretical Physics}, pages 43--110.
  Academic Press, 1967.

\bibitem{Kriz1990177}
Igor Kriz and Robin Thomas.
\newblock Clique-sums, tree-decompositions and compactness.
\newblock {\em Discrete Mathematics}, 81(2):177 -- 185, 1990.

\bibitem{Dissection}
R.~Lipton, D.~Rose, and R.~Tarjan.
\newblock Generalized nested dissection.
\newblock {\em SIAM Journal on Numerical Analysis}, 16(2):346--358, 1979.

\bibitem{PM_Little}
C.~Little.
\newblock An extension of {K}asteleyn's method of enumerating the 1-factors of
  planar graphs.
\newblock In {\em Combinatorial Mathematics}, LNCS, pages 63--72. 1974.

\bibitem{LATIN_K5}
B.~Reed and Z.~Li.
\newblock Optimization and recognition for {K}5-minor free graphs in linear
  time.
\newblock In {\em LATIN 2008: Theoretical Informatics}, pages 206--215. 2008.

\bibitem{DBLP:conf/gst/RobertsonS91}
N.~Robertson and P.~Seymour.
\newblock Excluding a graph with one crossing.
\newblock In {\em Graph Structure Theory}, pages 669--676, 1991.

\bibitem{Robertson200343}
N.~Robertson and P.~Seymour.
\newblock Graph minors. {XVI. E}xcluding a non-planar graph.
\newblock {\em Journal of Combinatorial Theory, Series B}, 89(1):43 -- 76,
  2003.

\bibitem{ThieraufPM}
S.~Straub, T.~Thierauf, and F.~Wagner.
\newblock Counting the number of perfect matchings in {K}5-free graphs.
\newblock {\em Electronic Colloquium on Comp. Complexity (ECCC)}, 21(79), 2014.

\bibitem{FisherTemperley}
H.~Temperley and M.~Fisher.
\newblock Dimer problem in statistical mechanics - an exact result.
\newblock {\em Philosophical Magazine}, 6(68):1061--1063, 1961.

\bibitem{DBLP:journals/tcs/Valiant79}
L.~Valiant.
\newblock The complexity of computing the permanent.
\newblock {\em Theor. C. Sci.}, pages 189--201, 1979.

\bibitem{DBLP:journals/siamcomp/Valiant08}
L.~Valiant.
\newblock Holographic algorithms.
\newblock {\em SIAM J. Comput.}, 37(5):1565--1594, 2008.

\bibitem{DBLP:journals/iandc/Vazirani89}
V.~Vazirani.
\newblock {NC} algorithms for computing the number of perfect matchings in
  ${K}_{3,3}$-free graphs and related problems.
\newblock {\em Inf. Comput.}, 80(2):152--164, 1989.

\end{thebibliography}
\endgroup
\end{document}